\PassOptionsToPackage{pdfpagelabels=false}{hyperref} 
\documentclass[final]{entcs}
\usepackage{entcsmacro}

\usepackage[T1]{fontenc}
\usepackage{tikz}
\usetikzlibrary{arrows.meta} 
\usepackage{anyfontsize}
\usepackage{wrapfig}
\usepackage{xspace}
\usepackage{mathrsfs}
\usepackage{mathtools}
\usepackage{stmaryrd}
\usepackage{xfrac}
\usepackage{cite}
\usepackage{mathpartir}

\def\Snospace~{\S{}}

\DeclarePairedDelimiter\abs{\lvert}{\rvert}%
\DeclarePairedDelimiter\norm{\lVert}{\rVert}%
\makeatletter
\let\oldabs\abs
\def\abs{\@ifstar{\oldabs}{\oldabs*}}
\let\oldnorm\norm
\def\norm{\@ifstar{\oldnorm}{\oldnorm*}}
\makeatother

\newcommand{\eps}{\varepsilon}
\newcommand{\ie}{i.e.\xspace}
\newcommand{\eg}{e.g.\xspace}
\newcommand{\llin}{\ensuremath{\lambda_1}\xspace}
\newcommand{\qstar}{\ensuremath{\mathsf{Q}^*}\xspace}

\newcommand{\K}{\mathcal{K}}
\newcommand{\N}{\mathcal{N}}

\newcommand{\meas}{\mathrm{\textbf{meas}}}

\newcommand{\trc}[1]{\mathrel{#1^*}}

\newcommand{\die}{\square}
\newcommand{\Lopen}{\left\lbrack}
\newcommand{\Lclose}{\right\rbrack}
\newcommand{\Lcons}{\mathop{:}}
\newcommand{\Lsep}{\mathpunct{,}}
\newcommand{\app}{{\mathbin{+\!\!\!+}}}
\newcommand{\nil}{\Lopen~\Lclose}
\newcommand{\A}{\mathcal{A}}
\newcommand{\Pow}{\mathcal{P}}
\newcommand{\Tree}{\mathcal{T}}
\newcommand{\liv}{\mathrm{Liv}}  
\newcommand{\dist}{d}  
\newcommand{\Dist}{\mathscr{D}}  
\newcommand{\List}{\mathscr{L}}
\newcommand{\Real}{\mathbb{R}}
\newcommand{\Det}{\mathrm{Det}}
\newcommand{\seman}[1]{\left \llbracket #1 \right \rrbracket}
\newcommand{\support}{\mathrm{supp}}
\newcommand{\Ra}{\mapsto}
\newcommand{\Rb}{\twoheadrightarrow}
\newcommand{\LRb}{\twoheadleftarrow\mathrel{\mkern-15mu}\twoheadrightarrow}
\newcommand{\La}{\mapsfrom}
\newcommand{\Lb}{\twoheadleftarrow}
\newcommand{\la}{\leftarrow}
\newcommand{\ra}{\mathbin{\rightarrow}}
\newcommand{\comm}{\mathrel{\dashv}}
\renewcommand{\sim}{\mathrel{\thicksim}}
\newcommand{\ssim}{\mathrel{\thickapprox}}
\renewcommand{\o}{\mathop{\cdot}}
\newcommand{\Pb}{\Rb_{P}}
\newcommand{\Pbl}{\Lb_{P}}
\newcommand{\rmod}{\mathord{/}} 
\newcommand{\LCR}{\mathrm{LC}}
\newcommand{\SN}{\mathrm{SN}}
\newcommand{\CR}{\mathrm{CR}}
\newcommand{\UN}{\mathrm{UN}}
\newcommand{\UTD}{\mathrm{UTD}}
\newcommand{\ULD}{\mathrm{ULD}}

\newcommand{\bang}[1]{!#1}
\newcommand{\bnfor}{~\mathbin{\mid}~}
\newcommand{\rulename}[1]{\ensuremath{\mbox{\textsc{#1}}}}
\newcommand{\rn}[1]{\rulename{#1}}
\newcommand{\lam}{\lambda}

\newcommand{\email}[1]{\href{mailto:#1} {\texttt{\normalshape #1}}}
\def\lastname{\textsc{D\'iaz-Caro and Mart\'inez}}

\begin{document}
\begin{frontmatter}
  \title{Confluence in Probabilistic Rewriting}
    \author[UNQ,ICC]{Alejandro D{\'\i}az-Caro\thanksref{proyectos}}
    \author[CIFASIS]{Guido Mart{\'\i}nez}

    \address[UNQ]{Universidad Nacional de Quilmes\\
      Roque S\'aenz Pe\~na 352. B1876BXD Bernal, Buenos Aires, Argentina }
      \address[ICC]{CONICET-Universidad de Buenos Aires. 
      Instituto de Investigaci\'on en Ciencias de la Computaci\'on\\
      Pabell\'on 1, Ciudad Universitaria. C1428EGA Buenos Aires, Argentina\\
      \email{adiazcaro@icc.fcen.uba.ar}
    }

    \address[CIFASIS]{CIFASIS-CONICET \& Universidad Nacional de Rosario\\
      Ocampo y Esmeralda, S2000EZP. Rosario, Santa Fe, Argentina \\
    \email{martinez@cifasis-conicet.gov.ar}
  }

    \thanks[proyectos]{Partially supported by projects ECOS-Sud A17C03 QuCa and PICT 2015-1208.}

\begin{abstract}
  Driven by the interest of reasoning about probabilistic programming languages,
  we set out to study a notion of uniqueness of normal forms for them. To provide a
  tractable proof method for it, we define a property of \emph{distribution confluence}
  which is shown to imply the desired uniqueness (even for infinite sequences of reduction) and further properties. We then
    carry over
  several criteria from the classical case, such as Newman's lemma, to simplify
  proving confluence in concrete languages.
  Using these criteria, we obtain simple proofs of confluence for \llin,
  an affine probabilistic $\lam$-calculus, and for \qstar, a quantum programming
  language for which a related property has already been proven in the literature.
\end{abstract}

\end{frontmatter}

\section{Introduction}

In the formal study of programming languages, modelling execution via a
small-step operational semantics is a popular choice.
Such a semantics is given by an \emph{abstract rewriting system} (ARS) which,
mathematically, is no more than a binary relation on terms specifying
whether one term can rewrite to another.
This relation is not required to be a function, and can thus allow for a
program to rewrite in two different ways.
In such a case, it is important that the different execution paths for
a given program reach the same final value (if any); thus guaranteeing
that any two determinisations of the semantics (\ie~strategies)
coincide on every program.

This correctness property is expressible at the level of
relations, and is known as \emph{uniqueness of normal forms} (UN):
any two irreducible terms reachable from a common starting point must be equal.
For non-trivial languages, such as the $\lam$-calculus,
it can be hard to prove this property directly.
%
Fortunately, the property of \emph{confluence} can serve as a proof
method for it, since it trivially implies UN and yet its proof
tends to be more tractable.
This was the approach followed by Church and Rosser in~\cite[Corollary 2]{ChurchRosser},
where UN and confluence were first proven for the $\lam$-calculus in 1936.
Nowadays, confluence is widely used to show the adequacy of operational
semantics in many kinds of programming languages.
%

Since some decades ago, there has been growing interest in
\emph{probabilistic} programming languages~\cite{Gordon,
Kozen1981} and, in particular, those where reduction itself is
stochastic.
Properties of such languages have been studied, in particular,
from a language-independent rewriting
approach~\cite{BournezHoyrupRTA03,BournezGarnier,BournezKirchner,KirkebyChristiansen}
and for specific calculi~\cite{DalLagoZorzi11,DalLago,DalLago2017,Jano2011,DiPierro}.
For example, taking the former approach, \cite{BournezGarnier} introduces a notion
of probabilistic termination of rewriting systems and provides
techniques to prove it, extending the classical ones to this
probabilistic setting.
Following the latter approach, \cite{DalLagoZorzi11} takes a
non-deterministic $\lam$-calculus and endows it with probabilistic
call-by-value and call-by-name operational semantics (in both
small-step and big-step style for each) and proves they simulate
each other via CPS translations, following a classic result by
Plotkin~\cite{PlotkinCPS}.

Our focus in this paper is the interaction between
probabilistic reductions and non-determinism, where the latter
comes from different possible reduction choices.
Such choices exist, for example, when a given program contains
two reducible subexpressions, each of which is probabilistic.
In that case, the program can take a step to two different
normalised distributions depending on the choice.
Given such non-determinism, a natural question
to ask is whether the result of a program
(which is a distribution of values~\cite{Kozen1981})
is not affected by the strategy, analogously
to UN for classical languages.
This is precisely the property we set out to study here,
developing an associated notion of confluence for it.

While this same property has already been studied in the
literature~\cite{Jano2011,DalLago}, it has only been done for concrete
languages, and so a language-independent study was previously lacking.
Further, the techniques employed are not immediately applicable to other
calculi.
There are also other studies of confluence in a probabilistic
setting at a more abstract level~\cite{BournezKirchner, KirkebyChristiansen},
but such notions are fundamentally different from what we previously
described, and of limited use for programming languages.

For a concrete example, let us take a hypothetical language for representing
die rolls, where
$\die$ represents an unrolled die which can reduce
to any element in $\{1,\ldots,6\}$ with equal probability
and `$-$' represents subtraction.
For a pair of dice
$(\die, \die)$, it should be allowed to choose which one to roll
first.
Rolling the first die can give in any term in the set
$\{(i,\die)\}_{i=1,6}$ with equal probability; and similarly for the
second.
When continuing the rolls, both alternatives provide the same
uniform distribution.
However, this could be not so: consider the term $(\lam x.~x-x)~\die$.
If the die is rolled before $\beta$-reducing, the result can
only be $0$, but if one $\beta$-reduces first, obtaining $(\die,
\die)$, the final result can be any number in $\{-5,\ldots,5\}$.
We shall later present a similar language and provide conditions to avoid
this discrepancy.

\textbf{Outline.}
In \autoref{sec:PARS} we introduce
probabilistic rewriting
and the problem of uniqueness of distributions.
In
\autoref{sec:RDC} we define a rewriting system over distributions
giving rise to our notion of confluence and prove it adequate.
In
\autoref{sec:Confluence} we derive criteria which simplify
the task of proving distribution for concrete languages.
In \autoref{sec:Limit} we extend some of our results to asymptotically terminating terms.
In \autoref{sec:examples} we prove confluence for two concrete calculi:
a simple probabilistic calculus dubbed $\llin$
and
the quantum lambda calculus \qstar~\cite{DalLago}.
Finally, in
\autoref{sec:Conclusions} we conclude, give some insights on future
directions, and analyse some related work.

\section{Probabilistic Rewriting}
\label{sec:PARS}

\subsection{Preliminaries}

We assume familiarity with abstract rewriting.
We adopt the terminology from \cite{ChurchRosser} and call
a sequence of expansions followed by a sequence of reductions
a \emph{peak}.
When the order is reversed, such a sequence is called a \emph{valley}.
The property of confluence can then be expressed as ``all peaks
have a valley''.
When an ARS $\A$ is confluent, we note it by $\A \models \CR$, and
similarly for $\UN$ and other properties.
The relation \emph{$R$ modulo $E$}, noted as
$R \rmod E$, is defined for any equivalence
relation $E$ as $E \o R \o E$~\cite{JouannaudKirchnerPOPL84,PetersonStickelJACM81}.
%
%

We denote by $\List(X)$ the type of finite lists with elements in $X$,
where `$\nil$', `$\Lcons$' and `$\app$' denote the empty list, the list
constructor, and list concatenation, respectively.
We also use the notation $\Lopen a \Lsep b \Lsep c \Lclose$ for
$a \Lcons b \Lcons c \Lcons \nil$.

We define $\Dist(A) = \List(\Real^+ \times A)$, used as an
explicit representation of (finitely-supported) distributions.
There is no further restriction on $\Dist(A)$.
In particular, any given element
might appear more than once, as in $\Lopen (\sfrac13, a) \Lsep (\sfrac12, b) \Lsep
(\sfrac16,a) \Lclose$.
For a point $(p,a)$ of the distribution, $p$ is called the \emph{weight}
and $a$ the \emph{element}.
The weight of a distribution is defined to be the sum of all weights of its points,
and can be any positive real number.
When a normalised distribution is required, we use
the type $\Dist_1(A)$, defined as the set of those $d \in \Dist(A)$ with
unit weight.
We abbreviate the distribution $\Lopen (p_1, a_1) \Lsep (p_2, a_2) \Lsep \ldots \Lsep (p_n, a_n) \Lclose$
by $\Lopen (p_i, a_i) \Lclose_i$, where $n$ should be clear from context.
We also write $\alpha D$ for the distribution obtained by scaling every weight
in $D$ by $\alpha$; that is, $\alpha \Lopen (p_i, a_i) \Lclose_i = \Lopen (\alpha p_i, a_i) \Lclose_i$.

For reasoning about equivalence of distributions, we define a relation `$\sim$'
as the congruence closure of following rules
(\ie~the smallest relation satisfying the rules and such that
$D \sim D'$ implies $E_1 \app D \app E_2 \sim E_1 \app D' \app E_2$, for any
$E_1, E_2$):
\[
    \arraycolsep=4pt
    \scalebox{0.94}{$
    \begin{array}{ccc}
        \multicolumn{1}{c}{
            \inferrule*[lab=Flip]{ }{
                        \Lopen (p,a) \Lsep (q,b) \Lclose
                        \sim
                 \Lopen (q,b) \Lsep (p,a) \Lclose
            }
        } &
        \inferrule*[lab=Join]{ }{
            \Lopen (p,a) \Lsep (q,a) \Lclose
                \sim
         \Lopen (p+q, a) \Lclose}
        &
        \inferrule*[lab=Split]{ }{
            \Lopen (p+q, a) \Lclose
                \sim
         \Lopen (p,a) \Lsep (q,a) \Lclose}
       \\
    \end{array}
    $}
\]
Note that $\sim$ is symmetric, since \rn{Join} is the inverse of \rn{Split}
and \rn{Flip} is its own inverse.
We distinguish some subsets of $\sim$ by limiting the rules that may be used.
We note by $S$ the congruence closure of \rn{Split},
by $(FJ)$ the congruence closure of both \rn{Flip} and \rn{Join},
and similarly for other subsets.

We call two distributions \emph{equivalent} when they are related
by the reflexive-transitive closure of $\sim$, noted `$\ssim$'.
Two distributions are equivalent, then, precisely when they assign
the same total weight to every element, irrespective of order and
multiplicity.

Arguably, using such a definition of distributions and equivalence
is cumbersome and less clear than using a more semantic one.
However, we feel this is
outweighed by the degree of rigor attained in later proofs (especially as we found some of
them to be quite error-prone).
As a secondary benefit, most of our development should be
straightforwardly mechanisable, since all equivalence steps are made explicit.

A tangible disadvantage of using lists is that only
\emph{finitely}-supported distributions can be represented.
This restriction (which is anyway lifted for infinite reduction sequences)
is present in other works as well (\eg~\cite{SelingerValiron,DiPierro})
and it does not seem severe for modelling programming languages.

\subsection{Probabilistic Abstract Rewriting Systems (PARS)}

To model the uncertainty in probabilistic rewriting, we cannot use a
simple relation between elements as used in ARS.
We must instead relate elements to \emph{distributions} of elements.
Further, to model elements such as $(\die, \die)$, we need to be able to
relate elements to several such distributions.
This motivates the following definition.
\begin{defn}
    \label{defn:pars}
    A probabilistic abstract rewriting system (PARS) is a pair
    $(A, {\Ra})$ where $A$ is a set and
    ${\Ra}$ a relation of type $\Pow(A \times \Dist_1(A))$
    (called the ``pointwise evolution relation'').
\end{defn}
It should be clear that every ARS is also a PARS
by taking ``Dirac'' distributions (\ie~normalised, single-point distributions).
We can provide a simple example of a PARS by extensionally listing ${\Ra}$, as is commonly done for ARS.

\begin{example}\label{exm:pars-1}
  Let $\A$ be the PARS given by
  \[
    \arraycolsep=2pt
    \begin{array}{lll@{\quad}lll@{\quad}lll@{\quad}lll}
      a &\Ra& \Lopen (\sfrac23, b) \Lsep (\sfrac13, c) \Lclose &
                                                       a &\Ra& \Lopen (\sfrac25, a) \Lsep (\sfrac35, d) \Lclose &
      b &\Ra& \Lopen (\sfrac12, c) \Lsep (\sfrac12, d) \Lclose &
                                                       c &\Ra& \Lopen (1, d) \Lclose \\
    \end{array}
  \]
  Here, $a$ is the only non-deterministic element.
  We call $d$ a \emph{terminal} element since it
  has no successor distributions.
  More significant examples are presented
  in \autoref{sec:examples}.
\end{example}


Execution in a PARS is a mixture of non-deterministic and probabilistic choices.
The first kind, corresponding to the $\Pow$ operator, occur when the machine chooses
a successor distribution for the current element.
The second kind, corresponding to the $\Dist_1$ operator, is a random
choice between the elements of the chosen successor distribution.
%
%
To model such execution, we
use the notion of \emph{computation tree}.

\begin{defn}
    \label{def:trees}
    Given a PARS $(A, {\Ra})$, we define the set of its (finite) ``computation trees''
    with root $a$ (noted $\Tree(a)$) inductively by the following rules.
    We also sometimes consider infinite computation trees, by taking the coinductively defined set instead.
    \[
        \inferrule*[]
          { }
          {a \in \Tree(a)}
    \qquad
    \inferrule*[]
          {a \Ra \Lopen (p_1, a_1) \Lsep \ldots \Lsep (p_n, a_n) \Lclose \\ t_i \in \Tree(a_i)}
          {[a ; (p_1,t_1) ; \ldots; (p_n,t_n) ] \in \Tree(a)}
      \]
\end{defn}

\noindent
\begin{minipage}[l]{\textwidth}
\setlength{\columnsep}{10pt}%
 \begin{wrapfigure}{r}{0pt}
   \scalebox{0.7}{
    \begin{tikzpicture}
        \draw(0,0)     node (a)  [] {$a$};
        \draw(-1,-1.5) node (b)  [] {$b$};
        \draw( 1,-1.5) node (c)  [] {$c$};
        \draw(-1.75,-3)   node (c') [] {$c$};
        \draw( -0.25,-3)   node (d)  [] {$d$};
        \draw( 1,-3)   node (d') [] {$d$};

        \draw[->,left]  (a) edge node {$\sfrac23$} (b);
        \draw[->,right] (a) edge node {$\sfrac13$} (c);
        \draw[->,left]  (b) edge node {$\sfrac12$} (c');
        \draw[->,right] (b) edge node {$\sfrac12$} (d);
        \draw[->,right] (c) edge node {$1$} (d');
    \end{tikzpicture}
   }
 \end{wrapfigure}
A graphical representation of a tree for the PARS in Example~\ref{exm:pars-1} is given to the right.
A tree in $\Tree(a)$ represents one possible (uncertain) evolution of the system
after starting on $a$.
There is no further assumption about trees:
in particular, if an element is expanded many times in a given tree,
different successor distributions may be used each time.
In other words, we do not assume a ``Markovian scheduler''~\cite{BournezGarnier,BournezHoyrupRTA03}.
\end{minipage}

Collecting the leaves of a tree, along with their accumulated
probabilities, gives rise to a normalised list distribution.
\begin{defn}
    \label{def:supp}
    The ``support'' of a tree $T$ is a normalised distribution, defined by:
    \[
        \begin{array}{rcl}
            \support(a) &=& \Lopen (1,a) \Lclose \\
            \support([a ; (p_1,t_1) ; \ldots; (p_n,t_n) ]) &=& p_1 \support(t_1) \app \ldots \app p_n \support(t_n) \\
        \end{array}
    \]
\end{defn}
When all the leaves of a tree are terminal elements, we call the tree
\emph{maximal} (as there is no proper supertree of it).
We can now state our property of interest.
\begin{defn}[UTD]
    A PARS $\A$ has ``unique terminal distributions''
    when for every $a$
    and $T_1, T_2 \in \Tree(a)$ maximal, we have
    $\support(T_1) \ssim \support(T_2)$.
\end{defn}
We stated before that proving UN (for ARS) directly is usually hard.
Since PARS subsume ARS, the same difficulties arise for proving UTD directly.
Therefore, we seek a property akin to confluence, providing a compositional and more
tractable proof method.

\section{Rewriting Distributions and Confluence}
\label{sec:RDC}

To arrive at a notion of confluence
we shall first define a rewriting over distributions (Definition~\ref{def:determinisation})
that is more liberal than that of computation trees (Definition~\ref{def:trees}).

\begin{defn}
    Given a PARS $\A = (A, {\Ra})$, we define the relation
    ${\Pb}$ (of type $\Pow(\Dist(A) \times \Dist(A))$)
    (called ``parallel evolution'') by the rules:
    \[
        \inferrule*[]
              {a \Ra A \\ ds \Pb ds'}
              {(p, a) \Lcons ds \Pb pA \app ds'}
        \qquad
        \inferrule*[]
              {ds \Pb ds'}
              {(p, a) \Lcons ds \Pb (p, a) \Lcons ds'}
        \qquad
        \inferrule*[]
              { }
              {\nil \Pb \nil}
    \]
\end{defn}

\noindent
Note that without using the first rule, this is just the identity relation on
distributions.
We note the subset of this relation where the first rule must be used at
least once in a step as ${\Pb^1}$, and call it \emph{proper} evolution.
The ${\Pb}$ relation can simulate
computation trees in this system, since it can be used to rewrite
their supports in the sense of Lemma~\ref{lem:tree-sim}.

\begin{defn}
    We call a relation ${\ra}$ ``compositional'' when,
    if $D_1 \ra E_1$ and $D_2 \ra E_2$, then $\alpha D_1 \app \beta D_2 \ra
    \alpha E_1 \app \beta E_2$ for all $\alpha, \beta \in \Real^+$.
\end{defn}
\begin{lemma}
  \label{lem:tree-sim}
  If $T \in \Tree(a)$, then $\Lopen (1,a) \Lclose \trc\Pb \support(T)$
\end{lemma}
\begin{proof}
    First, note that ${\Pb}$ is compositional.
    The result then follows by induction on $T$, using compositionality.
\end{proof}

We now define an ARS over distributions, combining both parallel evolution and
equivalence steps.
Our definition of confluence for a PARS $\A$ is then simply
the usual confluence
of that relation.

\begin{defn}\label{def:determinisation}
    Given a PARS $\A = (A, {\Ra})$, we define an associated ARS $\Det(\A)$
    (called the ``determinisation'' of $\A$)
    over the set $\Dist(A)$ by the relation ${\Rb} = ({\Pb} \cup {\ssim})$.
\end{defn}

\begin{defn}[Distribution confluence]
    \label{def:conf}
    We say a PARS $\A$ is ``distribution confluent'' (or simply ``confluent'')
    when $\Det(\A)$ is confluent in the classical sense.
\end{defn}

Reduction in $\Det(\A)$ is more liberal than the expansion of trees,
since it allows for ``partial'' evolutions.
Indeed, if $a \Ra D$,
then $\Lopen (1,a) \Lclose \Rb \Lopen (\sfrac12, a) \Lsep (\sfrac12, a) \Lclose \Rb
\sfrac12 D \app \Lopen (\sfrac12, a) \Lclose$.
Nevertheless, Lemma~\ref{lem:cr-utd} shows that its confluence implies UTD.

\begin{lemma}
  \label{lem:terminal-evolution}
  If $D_1$ is terminal and $D_1 \trc\Rb D_2$, then $D_1 \ssim D_2$ and $D_2$ is
  terminal.
\end{lemma}
\begin{proof}
    It is clear, from the definition of ${\Pb}$, that if $D_1$ is terminal
    and $D_1 \Pb D'$, then $D_1 = D'$ (that is, exactly equal).
    The result then follows by induction on the number of steps, and the
    transitivity and reflexivity of ${\ssim}$.
\end{proof}

\begin{lemma}
    \label{lem:cr-utd}
    If $\A \models \CR$, then $\A \models \UTD$.
\end{lemma}
\begin{proof}
  Take $T_1, T_2 \in \Tree(a)$ maximal. We have from
  Lemma~\ref{lem:tree-sim} that $\support(T_2) \trc\Lb \Lopen (1,a) \Lclose \trc\Rb
  \support(T_1)$.
  By confluence, there must exist $C$ such that $\support(T_2) \trc\Rb C \trc\Lb
  \support(T_1)$.
  Since $T_1, T_2$ are maximal, their supports are terminal.
  Then, from Lemma~\ref{lem:terminal-evolution}, we get that
  $\support(T_2) \ssim C \ssim \support(T_1)$, as needed.
\end{proof}

Furthermore, beyond UTD, distribution confluence implies that diverging computations
(with no terminal distribution)
can also be joined.
As a consequence of that, confluence gives a neat method for proving
the consistency of the equational theory induced by $\Rb$, as long as two distinct
terminal elements exist.
\begin{lemma}
  If $D_1, D_2$ are terminal distributions and $\A$ is confluent, then
  $D_1 \trc\LRb D_2$ if and only if $D_1 \ssim D_2$.
\end{lemma}
\begin{proof}
    The way back is trivial, so we detail the way forward.
    From confluence (repeatedly),
    $D_1$ and $D_2$ must have a common reduct.
    The result then follows from Lemma~\ref{lem:terminal-evolution}.
\end{proof}

Then, if $a$ and $b$ are distinct terminal elements,
it follows that ${\Rb}$-convertibility is a consistent theory
as $\Lopen (1,a) \Lclose \not\ssim \Lopen (1,b) \Lclose$.
Summarizing, in a confluent PARS, reasoning about equivalence of programs is simplified
and there is a strong consistency guarantee about convertibility, much like
in the classical case.

Readers familiar with rewriting modulo
equivalence\cite{JouannaudKirchnerPOPL84,PetersonStickelJACM81} may
wonder why we are not studying confluence modulo $\ssim$, or the stronger
Church-Rosser property modulo $\ssim$.
It turns out both of these are too strong for our purposes.
The following system:
\[
    a \Ra \Lopen (\sfrac12, a) \Lsep (\sfrac12, b) \Lclose
\]
should undoubtedly be considered confluent due to being deterministic;
but in it we can form the following diagram
\[
    \Lopen (\sfrac12, a) \Lsep (\sfrac12, b) \Lclose \Pbl
    \Lopen (1,a) \Lclose
    \ssim
    \Lopen (\sfrac13, a) \Lsep (\sfrac23, a) \Lclose
\]
which cannot be closed by $\trc\Pb \o \ssim \o \trc\Pbl$, due to a factor
of $3$ present in one side and not the other.
Thus this system is neither confluent modulo $\ssim$ nor $\CR$ modulo
$\ssim$, yet it is clearly distribution confluent.
We therefore study the strictly weaker distribution
confluence, which is strong enough for our purposes and easier to reason
about.

\section{Proving Distribution Confluence}\label{sec:Confluence}
\subsection{Introduction}

In the previous section, we introduced our definition of confluence and
argued about its correctness.
For it to be useful in practice, it should also be amenable to be proven. In this
section we provide several simplified criteria for this task, obtaining
analogues to the most usual methods for proving classical confluence.

Since distribution confluence is simply the classical confluence
of ${\Rb}$, all existing classical criteria (\eg~the diamond property or
Newman's lemma) are valid without modification. However, they are not very
useful.
One issue is that one needs to consider all distributions (instead of single elements) and
the presence of equivalence steps in the peaks.
Also, $\Det(\A)$ is never strongly (or even weakly) normalising,
so Newman's lemma is useless here.
With respect to the diamond
property, consider a system with $a \Ra D$, then the following reductions are
possible:
\[
    \Lopen (1,a) \Lclose
    \Lb
        \Lopen (\sfrac12, a) \Lsep (\sfrac12,a) \Lclose
    \Rb
        \sfrac12 D \app \Lopen (\sfrac12, a) \Lclose
\]
and these two distributions cannot in general be joined in a single step,
even if $a$ has no other successor distribution.

Thus, a priori, it seems as if distribution confluence is even harder
to prove than classical confluence.
To relieve that, we shall prove various syntactic lemmas about the
${\Rb}$ relation, allowing us to decompose it into more manageable
forms.
We then show how we can limit our reasoning to Dirac distributions,
ignore equivalence steps in the peaks and allow to use them freely in
the valleys.
Lastly, we carry over classical criteria for confluence into our
setting, such as the aforementioned diamond property and Newman's lemma.

\subsection{Decomposing the \texorpdfstring{${\Rb}$}{->>} relation}

Since both ${\Pb}$ and ${\ssim}$ are reflexive, $({\Pb} \cup {\ssim})^*$ coincides with $({\Pb}
\rmod{\ssim})^*$.
Thus, since confluence is a property over the reflexive-transitive closure
of a relation, it suffices to study the confluence of ${\Pb} \rmod {\ssim}$, where
equivalence steps do not have a ``cost'', but are pervasive (as in rewriting modulo equivalence).%
%

Given the precise syntactic definition for both relations, we can prove by
analysing the reductions that any step of ${\Pb} \rmod {\ssim}$ can be made by
first splitting, then evolving, and then joining back elements, as
Lemma~\ref{lem:decomp-1} states.
We first introduce the following notion of commutation.%
    \footnote{%
    Note that this is not the usual
    notion of commuting relations,
    defined as $S^{-1} \o R \subseteq R \o S^{-1}$,
    which is a symmetric property and could be described as ``parallel''.}

\noindent
\begin{minipage}[l]{0.86\textwidth}
\begin{defn}[Sequential commutation]
    We say that a relation $R$ ``commutes over'' $S$ when $S \o R
    \subseteq R \o S$, and note it as $R \comm S$.
    The property can be expressed by the diagram on the right.
    Intuitively, it means that $R$ can be ``pushed'' before $S$.
\end{defn}
\end{minipage}
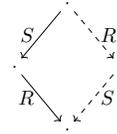
\begin{minipage}[r]{0.13\textwidth}
        \scalebox{0.7}{
            \begin{tikzpicture}[baseline=0pt,yscale=0.8]
                \draw(0,0)    node (a) [] {.};
                    \draw(-1,-1.5)  node (b) [] {.};
                    \draw( 1,-1.5)  node (c) [] {.};
                    \draw( 0,-3)  node (d) [] {.};

                    \draw[->,left ]  (a) edge node {$S$} (b);
                    \draw[->,right,dashed]  (a) edge node {$R$} (c);
                    \draw[->,left]  (b) edge node {$R$} (d);
                    \draw[->,right,dashed]  (c) edge node {$S$} (d);
            \end{tikzpicture}}
\end{minipage}

A key property of sequential commutation is that if $R \comm S$,
then $(R \cup S)^* = R^* \o S^*$. It is also preserved when taking
the $n$-fold composition (\ie~``$n$ steps'') or reflexive-transitive closures on each side.
We now prove some commutations relating evolution and equivalence
steps (the last one needs some ``administrative'' steps).

\begin{lemma}\label{lem:comm-1}
    We have ${\Pb} \comm (FJ)^*$; $S \comm (FJ)^*$ and
    ${\Pb} \o S \subseteq S \o {\Pb} \o (FJ)^*$.
\end{lemma}
\begin{proof}
    By induction on the shape of the reductions.
\end{proof}
\begin{lemma}
    \label{lem:decomp-1}
    The relations $({\Pb} \rmod {\ssim})$ and $S^* \o {\Pb} \o (FJ)^*$
    coincide.
\end{lemma}
\begin{proof}
    The backwards inclusion is trivial,
    so we detail only the forward direction.
    By making use of the second commutation in Lemma~\ref{lem:comm-1}
    we get that ${\ssim} = S^* \o (FJ)^*$.
    Thus, we need to show
    $S^* \o (FJ)^* \o {\Pb} \o S^* \o (FJ)^* \subseteq S^* \o {\Pb} \o (FJ)^*$.
    The proof then proceeds by using the other two commutations to
    reorder the relations.
\end{proof}

Further, this equivalence extends to $n$-fold compositions and therefore to the
reflexive-transitive closure.

\begin{lemma}
    \label{lem:decomp-n}
    The relations $({\Pb} \rmod {\ssim})^n$ and $S^* \o {\Pb^n} \o (FJ)^*$
    coincide.
\end{lemma}
\begin{proof}
    By induction on $n$, and using the previous lemma and commutations.
\end{proof}

\subsection{Simplifying diagrams}
With the previous decompositions, we can now prove a very generic result about
diagram simplification with a specific root $D$, which then easily generalizes
to the whole system.

\noindent
\begin{minipage}[l]{0.86\textwidth}
\begin{defn}
    We say a pair of relations $(\gamma, \delta)$ ``closes'' another
    pair $(\alpha, \beta)$ ``on $a$''
    if whenever $b \la_\alpha a \ra_\beta c$
    then there exists $d$ such that $b \ra_\gamma d \la_\delta c$.
    The diagram for the property can be seen on the right.
    When this occurs for all $a$, we simply say ``$(\gamma, \delta)$ closes $(\alpha, \beta)$''.
    Note that ${\ra}$ is confluent precisely when $({\trc\ra}, {\trc\ra})$ closes $({\trc\ra}, {\trc\ra})$.
\end{defn}
\end{minipage}
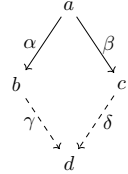
\begin{minipage}[r]{0.13\textwidth}
    \scalebox{0.7}{
      \begin{tikzpicture}[baseline=0pt]
        \draw(0,0)    node (a) [] {$a$};
        \draw(-1,-1.5)  node (b) [] {$b$};
        \draw( 1,-1.5)  node (c) [] {$c$};
        \draw( 0,-3)  node (d) [] {$d$};

        \draw[->,left ]  (a) edge node {$\alpha$} (b);
        \draw[->,right]  (a) edge node {$\beta$} (c);
        \draw[->,left, dashed]  (b) edge node {$\gamma$} (d);
        \draw[->,right,dashed]  (c) edge node {$\delta$} (d);
      \end{tikzpicture}
    }
\end{minipage}

\begin{defn}
    We call a relation $\ra$ ``local'',
    when if $\alpha D_1 \app \beta D_2 \ra E$, then there exist
    $E_1, E_2$ such that $E = \alpha E_1 \app \beta E_2$
    and $D_i \ra E_i$.
    (Note that ${\Pb}$ and $S$ are local).
\end{defn}

\begin{thm}
    \label{thm:decomp}
    Let $\alpha, \beta$ be local relations
    and $\gamma, \delta$ compositional relations.
    If $(\gamma \rmod {\ssim}, \delta \rmod {\ssim})$
    closes $(\alpha, \beta)$ for the Dirac distributions of $D$,
    then $(\gamma \rmod {\ssim}, \delta \rmod {\ssim})$
    closes $(S^* \o \alpha \o (FJ)^*, S^* \o \beta \o (FJ)^*)$ for $D$.
\end{thm}
\begin{proof}
    We give a sketch of the proof, more details can be found in \cite{Martinez17}.
    We need to close
    $(S^* \o \alpha \o (FJ)^*, S^* \o \beta \o (FJ)^*)$.
    First, note that closing $(S^* \o \alpha, S^* \o \beta^*)$ is enough since we can
    revert the $(FJ)^*$ steps with $(FS)^*$ steps.
    Now, since $\alpha$ and $\beta$ are local, $S^* \o \alpha$ and $S^* \o \beta$ are as well.
    Thus, we can limit ourselves to closing the Dirac distributions of $D$,
    and combine the reductions since $\gamma, \delta$ are compositional.
    We now need to close $(S^* \o \alpha, S^* \o \beta)$ when starting from some $\Lopen(1,a)\Lclose$.
    Note that the left (right) branch is then of the form
    $p_1 D_1 \app \dots \app p_n D_n$
    ($q_1 E_1 \app \dots \app q_m E_m$), where $a$
    reduces via $\alpha$ ($\beta$) to each $D_i$ ($E_j$).
    We can apply our hypothesis to get a $C_{i,j}$ closing each $D_i, E_j$.
    By first splitting each branch appropriately, we can close them
    in $p_1 q_1 C_{1,1} \app \ldots \app p_1 q_m C_{1,m} \app \ldots \app p_n q_m C_{n,m}$.
\end{proof}

From this theorem, we get as corollaries several simplified criteria for
confluence,
applicable at the level
of a particular distribution or to the whole system.

\begin{crit}[Dirac confluence]
    \label{cor:elementVsDist}
    If for every element $a$ of $D$ and distributions $E, F$ such that
    $E \trc\Pbl \Lopen (1,a) \Lclose \trc\Pb F$ there is a $C$
    such that
    $E \trc\Rb C \trc\Lb F$, then $D$ is confluent.
\end{crit}
\begin{proof}
    A corollary of Theorem~\ref{thm:decomp}, taking
    $\alpha = \beta = \gamma = \delta = {\trc\Pb}$.
\end{proof}

\begin{crit}[Semi-confluence]
    \label{crit:semi}
    If for every element $a$ of $D$ and distributions $E, F$ such that
    $a \Ra E$ and $\Lopen (1,a) \Lclose \trc\Pb F$ there is a $C$
    such that
    $E \trc\Rb C \trc\Lb F$, then $D$ is semi-confluent for ${\Pb} \rmod {\ssim}$.
\end{crit}
\begin{proof}
    A corollary of Theorem~\ref{thm:decomp}, taking
    $\alpha = {\Pb}$ and $\beta = \gamma = \delta = {\trc\Pb}$.
\end{proof}

\begin{crit}[Diamond property]
    \label{crit:diam}
    If for every element $a$ of $D$ and distributions $E, F$ such that
    $E \La a \Ra F$ there is a $C$
    such that
    $E \Rb_{P \rmod {\ssim}} C \Lb_{P \rmod {\ssim}} F$, then $D$ has the diamond property for ${\Pb} \rmod {\ssim}$.
\end{crit}
\begin{proof}
    A corollary of Theorem~\ref{thm:decomp}, taking
    $\alpha = \beta = \gamma = \delta = {\Pb}$.
\end{proof}

Note that in all these criteria, we need not consider
any equivalence in the peak, and can use them freely in the valley, both before and after evolving.
Also, proving any of these criteria for every element $a$ entails the
confluence of the system.

Another common tool for proving confluence is switching
the relation to another one with equal reflexive-transitive closure (and thus an
equivalent confluence) but which might be easier to analyse.
For distribution confluence, a similar switch is allowed,
slightly simplified by Lemma~\ref{lem:helper}.
\begin{defn}
  Given ${\Ra_1}$ and ${\Ra_2}$ over the same set $A$,
  if for every $a \Ra_1 D$, we have $\Lopen (1,a) \Lclose \trc\Rb_2 D$ we
  say that ${\Ra_1}$ is \emph{simulated by} ${\Ra_2}$.
\end{defn}
\begin{lemma}\label{lem:helper}
    If ${\Ra_1}$ is simulated by ${\Ra_2}$, then ${\Rb_1} \subseteq {\trc\Rb_2}$.
    If both relations simulate each other, then ${\trc\Rb_1} = {\trc\Rb_2}$,
    and their confluences are equivalent.
\end{lemma}
\begin{proof}
    The first part follows by case analysis on the reduction ${\Rb_1}$.
    The second part is trivial.
\end{proof}

\subsection{Newman's lemma}

Newman's lemma~\cite{Newman} states that, for a strongly normalising system,
local confluence and confluence are equivalent,
yet we have remarked previously that ${\Pb}$ is never a strongly normalising relation.
To get an analogue to Newman's lemma, we thus provide a specialized notion of strong normalisation.

\begin{defn}
    \label{def:chain}
    A infinite sequence $D_i$ such that
    $D_1 \ra D_2 \ra D_3 \ra \cdots$
    is called an ``infinite ${\ra}$-chain'' (of root $D_1$).
\end{defn}

\begin{defn}[SN]
  We call a distribution $D$ ``strongly normalising'' when there is no
  infinite ${\Pb^1}$-chain
  of root $D$.%
  \footnote{
      Note that infinite $({\Pb^1} \rmod {\ssim})$-chains always
      exist because of partial evolution.}
  We call a PARS strongly normalising when every distribution
  is strongly normalising.
\end{defn}

There are indeed systems which do satisfy this
requirement, and it is intuitively what one would expect. Now a
probabilistic analogue to Newman's lemma can be obtained, following
a proof style very similar to that of \cite{Huet}.

\begin{defn}[LC]
  We say that a distribution $D$ is ``locally confluent'' when $E
  \Pbl^1 D \Pb^1 F$ implies that there exists $C$ such that $E \trc\Rb C \trc\Lb F$.
\end{defn}

Note that strong normalisation over Dirac distributions implies it for all distributions,
and likewise for local confluence.

\begin{lemma}[Newman's]
    \label{lem:newman}
    If $\A \models \LCR$ and $\A \models \SN$, then $\A \models \CR$.
\end{lemma}
\setlength{\columnsep}{10pt}%
  \begin{wrapfigure}{r}{0pt}
    \scalebox{0.77}{
        \begin{tikzpicture}[yscale=0.7]
      \draw(0,0)    node (M) [] {$D$};
      \draw(-1,-1.5)  node (M1') [] {$E'$};
      \draw( 1,-1.5)  node (M2') [] {$F'$};
      \draw(-2,-3)  node (M1) [] {$E$};
      \draw( 2,-3)  node (M2) [] {$F$};
      \draw( 0,-3)  node (M') [] {$C$};
      \draw( 0,-1.5)  node (LC) [] {$\LCR$};
      \draw(-1,-3)  node (LC) [] {IH};
      \draw(0.5,-3.75)  node (LC) [] {IH};

      \draw(-1,-4.5)  node (C1) [] {$C'$};
      \draw( 0,-6)  node (C2) [] {$C''$};

      \draw[->,left]  (M)      edge node {\small$P^1$} (M1');
      \draw[->,right] (M)      edge node {\small$P^1$} (M2');
      \draw[->,left]  (M1')    edge node {\small$P^{1*}$} (M1);
      \draw[->,right] (M2')    edge node {\small$P^{1*}$} (M2);
      \draw[->>,dashed] (M1')  edge node[sloped,pos=.9,below=0] {*} (M');
      \draw[->>,dashed] (M2')  edge node[sloped,pos=.9,below=0] {*} (M');
      \draw[->>,dashed] (M1)   edge node[sloped,pos=.9,below=0] {*} (C1);
      \draw[->>,dashed] (M')   edge node[sloped,pos=.9,below=0] {*} (C1);
      \draw[->>,dashed] (M2)   edge node[sloped,pos=.95,below=0] {*} (C2);
      \draw[->>,dashed] (C1)   edge node[sloped,pos=.9,below=0] {*} (C2);
    \end{tikzpicture}
    }
  \end{wrapfigure}
\begin{proof}
    We shall prove, by well-founded induction over ${\Pb^1}$,
    that every distribution is confluent.
    For a given distribution, it suffices to show that
    that any
    peak of proper evolutions can be closed by ${\trc\Rb}$.
    Then, by Corollary~\ref{cor:elementVsDist} (and since ${\trc\Pb} = {\Pb^{1*}}$), confluence follows.
    We want to close a diagram of shape $E \Pbl^{1*} D \Pb^{1*} F$.
    If either of the branches is zero steps long,
    then we trivially conclude. If not, we can form the diagram on the right,
    completing the proof by local confluence and the induction hypotheses for $E'$ and $F'$.
\end{proof}

\section{Limit Distributions}
\label{sec:Limit}

In classical abstract rewriting, an element can either be non-normalising,
weakly normalising or strongly normalising (corresponding to the situations where
it \emph{will not}, \emph{may}, and \emph{will} normalise, respectively).
In probabilistic rewriting, the story is not as simple. Consider the
following PARS, where $b$ is a terminal element:
\[
    a \Ra \Lopen (\sfrac12, a) \Lsep (\sfrac12, b) \Lclose
\]

Is $a$ normalising? One could say ``no'' since, indeed, it does not have a finite maximal computation tree,
as there is always some probability for the system to be in the non-terminal $a$
state.
However, such a probability will be made arbitrarily small by taking
sufficient steps, and the distribution $\Lopen (1,b) \Lclose$ is reached
\emph{in the limit}.
In this case $a$ is called \emph{almost surely terminating}\cite{BournezGarnier}.
Certainly, a desirable fact is that almost-surely-terminating elements have a
unique final distribution.
We will prove that distribution confluence guarantees such uniqueness.

We first introduce a notion of distance between
\emph{mathematical distributions},
\ie~normalised functions of type $A \ra [0,1]$.
The reason we move away from list distributions is to allow for
infinitely-supported limit distributions.
We note with $\seman{D}$ the mathematical
distribution obtained from the list distribution $D$ (with the expected definition).
We also extend definitions over mathematical distributions to list distributions
by applying $\seman{-}$ where appropriate.

\begin{defn}
  Given $D, E$ mathematical distributions, we define the distance between them as
  $\dist(D,E) = \sum_{a \in A} \abs{D(a) - E(a)}$.
\end{defn}

\begin{defn}[Limit of a sequence]
  Given an infinite sequence of mathematical distributions $D_0, D_1, \ldots$ we say that $L$ is a
  limit for the sequence if for every $\eps > 0$, there exists $N > 0$ such that
  for all $i \ge N$, $\dist(D_i, L) < \eps$.
\end{defn}

Note that this
distance is the $L^1$ distance and the definition of limit is the usual one
for metric spaces.
It is then well known that limits for a given sequence are unique.
%
We are interested in limits composed of terminal elements, representing a
distribution of values.
For that, the following definition is useful.

\begin{defn}
  For a mathematical distribution $D$, we define its ``liveness'' as the sum of
  weights for non-terminal elements. That is,
  $
      \liv(D) = \sum_{a \in \mathrm{dom}(\Ra)} D(a)
  $
\end{defn}

Note that the liveness of a list distribution cannot increase by evolution,
and that $\liv(D) = 0$ iff $D$ is terminal.
Moreover, since the normalised part of a distribution cannot evolve,
liveness provides an upper bound on the possible distance to be attained
by evolution, as the following lemma states.

\begin{lemma}
    \label{lem:liv-dist}
    If $D \trc\Rb E$, then
    $\dist(D, E) \le 2 \cdot \liv(D)$.
\end{lemma}
\begin{proof}
  By Lemma~\ref{lem:decomp-n} there exist
  $D'$ and $E'$ such that $D \ssim D' \trc\Pb E' \ssim E$.
  Because of the equivalences, it suffices to show the result for $D'$ and $E'$.
  Assume, without loss of generality, that $D' = D_l \app D_t$, where all
  elements of $D_l$ are not terminal, and all those of $D_t$ are.
  Since parallel evolution is local and terminal elements cannot evolve, we have
  that $E' = E'' \app D_t$ for some $E''$.
  Then, $\dist(D', E')$ is simply $\dist(D_l, E'')$.
  Note that $\liv(D')$ is the weight of $D_l$ and of $E''$.
  Since distance is bounded by total weight, it follows that it is at most
    $2 \cdot \liv(D') = 2 \cdot \liv(D)$.
    %
\end{proof}

Now, we can extend our notion of uniqueness of terminal distributions to limit distributions
of terminal elements, accounting for an infinite sequence of reductions.
\begin{defn}[ULD]
    A PARS $\A$ has ``unique limit distributions'' when for every
    $D$ that is the root of
    two infinite \mbox{${\Rb}$-chains} $E_i$ and $F_j$ with respective limits $E_\infty$ and
    $F_\infty$ terminal distributions,
    then $E_\infty = F_\infty$.
\end{defn}
\begin{lemma}
    If $\A \models \CR$, then $\A \models \ULD$.
\end{lemma}
\noindent
\begin{minipage}[l]{\textwidth}
\setlength{\columnsep}{10pt}%
\begin{wrapfigure}{r}{0pt}
    \scalebox{0.8}{
        \begin{tikzpicture}[yscale=0.9]
            \draw(0,0)    node (D) [] {$D$};
            \draw(-1,-1.5)  node (Ei) [] {$E_i$};
            \draw( 1,-1.5)  node (Fj) [] {$F_j$};
            \draw(-2,-3)  node (E) [] {$E_\infty$};
            \draw( 2,-3)  node (F) [] {$F_\infty$};
            \draw( 0,-3)  node (C) [] {$C$};

            \draw[->>,sloped,left]  (D)   edge node[sloped,pos=.9,below=0] {*} (Ei);
            \draw[->>,sloped,right] (D)   edge node[sloped,pos=.9,below=0] {*} (Fj);

            \draw[->>]       (Ei)  edge node[sloped,pos=.4,above] {\tiny$[0,\sfrac\eps3)$} node[sloped,pos=.9,below=0] {*} (C);
            \draw[->>]       (Fj)  edge node[sloped,pos=.4,above] {\tiny$[0,\sfrac\eps3)$} node[sloped,pos=.9,below=0] {*} (C);
            \draw[->>,dotted] (Ei) edge node[sloped,pos=.4,above] {\tiny$[0,\sfrac\eps6)$} node[sloped,pos=.9,below=0] {*} (E);
            \draw[->>,dotted] (Fj) edge node[sloped,pos=.4,above] {\tiny$[0,\sfrac\eps6)$} node[sloped,pos=.9,below=0] {*} (F);
    \end{tikzpicture}
    }
  \end{wrapfigure}
\begin{proof}
    Take $\eps > 0$.
    %
  By the definition of limit, we know there are $i,j$ such that
    $\dist(E_i, E_\infty) < \sfrac\eps6 $ and $ \dist(F_j, F_\infty) < \sfrac\eps6$.
  Since $E_\infty$ and $F_\infty$ are terminal, $\liv(E_i)$ and
  $\liv(F_j)$ must be less than $\sfrac\eps6$.
  The distributions $E_i$ and $F_j$ are reachable by a finite amount of
  ${\Rb}$ steps, so by confluence there exists a distribution $C$ such that $E_i
  \trc\Rb C \trc\Lb F_j$.
  From Lemma~\ref{lem:liv-dist}, we get that $\dist(E_i, C) < \sfrac\eps3$ and
  likewise for $F_j$.
  From these four bounds and the triangle inequality
  we get
  that $\dist(E_\infty, F_\infty) < \eps$.
  Since this is the case for any positive $\eps$, $\dist(E_\infty, F_\infty)$ must be exactly
  $0$, and therefore $E_\infty = F_\infty$.
\end{proof}
\end{minipage}

\section{Case Studies}\label{sec:examples}


\subsection{An affine probabilistic \texorpdfstring{$\lam$}{lambda}-calculus: \texorpdfstring{$\llin$}{lambda1}}
\label{sec:llin}

In our introductory example, we used the term $(\lam x.~x-x)~\die$ as an
example of a non-confluent computation.
%
%
There seem to be three ingredients needed for this failure of confluence
of a term $(\lam x.M)N$:
(1) $x$ appears free more than once in $M$
(2) $N$ has a non-Dirac terminal distribution
(3) both call-by-name and call-by-value reductions are possible.

In this section we define a probabilistic $\lam$-calculus,
dubbed $\llin$, that prevents the combination of these three features by
providing two kinds of abstractions, one restricting duplication and one restricting evaluation order.%
\footnote{For more expressivity, a third kind
without either restriction, but forbidding
probabilistic arguments, could be added. We do not deem this as interesting for the scope of this paper.}
We show $\llin$ to be confluent (by a diamond property), giving evidence
that little more than affinity of probabilistic arguments is required
to achieve a confluent probabilistic programming language.

The calculus is heavily based on the one defined in~\cite{Simpson}. The set of
\emph{pre}-terms is given by the following grammar
\[
    M, N ::= x \bnfor M N
               \bnfor \lam x.M
               \bnfor \lam!x.M
               \bnfor \bang{M}
               \bnfor M \oplus_p N
\]
where the main novelty is the probabilistic choice operator $\oplus_p$, for any
real number $p$ in the open interval $(0,1)$.
For an abstraction $\lam x. M$, $x$ must be \emph{affine} in $M$, that
is, $M$ can have at most one free occurrence of $x$.
If there is exactly one such occurrence, we say $x$ is \emph{linear} in $M$.
Banged abstractions ($\lam!$) have no such restriction.
Affinity is enforced by a well-formedness judgment, whose definition
is straightforward and which we thus omit.
We work only with well-formed pre-terms, which form the set of terms.
For the sake of brevity, given a term $M$ and distribution $D = \Lopen (p_i, N_i) \Lclose$
we use the notation $MD$ to represent the distribution $\Lopen (p_i, M N_i) \Lclose$,
and similarly for all other syntactic constructs.

\begin{figure}
    \[
        \def\arraystretch{2.9}
    \arraycolsep=10pt
    \begin{array}{cc}
        \inferrule*[lab=R-$\beta$]
                { }
        {(\lam x. M)N \Ra \Lopen (1, M[N/x]) \Lclose}
        &
        \inferrule*[lab=R-$\beta!$]
                                { }
        {(\lam! x. M)\bang{N} \Ra \Lopen (1, M[N/x]) \Lclose}
                                \\
        \multicolumn{2}{c}{
            \inferrule*[lab=R-$\oplus$]
                    { }
            {M \oplus_p N \Ra \Lopen (p, M) \Lsep (1\!-\!p, N) \Lclose }
        }
    \end{array}
\]
    \[
        \def\arraystretch{3.0}
    \arraycolsep=10pt
    \begin{array}{ccc}
        \inferrule*[lab=R-$\oplus$-L]
                        {M \Ra D}
            {M \oplus_p N \Ra D \oplus_p N }
        &
        \inferrule*[lab=R-$\oplus$-R]
                        {N \Ra D}
        {M \oplus_p N \Ra M \oplus_p D }
        &
        \inferrule*[lab=R-AppL]
                        {M \Ra D}
        {MN \Ra DN}
        \\
        \inferrule*[lab=R-AppR]
                        {N \Ra  D}
        {MN \Ra  MD}
        &
        \inferrule*[lab=R-$\lam$]
                 {M \Ra D}
        {\lam x.M \Ra \lam x.D}
        &
        \inferrule*[lab=R-$\lam!$]
                        {M \Ra D}
        {\lam! x.M \Ra \lam! x.D}
    \end{array}
\]
    \caption{Full semantics for \llin}
    \label{fig:llin-semantics}
\end{figure}

The operational semantics is provided as a PARS in
Fig.~\ref{fig:llin-semantics}.
Terms of the form $\bang{M}$ do not reduce and are called
\emph{thunks}.
A banged abstraction can only $\beta$-reduce when applied to a thunk.
This effectively implies that banged abstractions follow a fixed strategy
(which is, morally, call-by-value until the argument is reduced to a thunk and call-by-name afterwards).%
\footnote{We are thus adopting ``surface reduction'' only
since ``internal reductions'' \cite{Simpson} hinder confluence.}

To prove the diamond property for $\llin$, we first need two substitution lemmas.
When $D = \Lopen (p_i, a_i) \Lclose_i$, we write $D[M/x]$
for the distribution $\Lopen (p_i, a_i[M/x]) \Lclose_i$.
Similarly, $M[D/x]$ denotes $\Lopen (p_i, M[a_i/x]) \Lclose_i$.

\begin{lemma}\label{lem:lin-subst-out}
    If $M \Ra D$, then $M[N/x] \Ra D[N/x]$.
\end{lemma}
\begin{proof}
    By induction on $M \Ra D$.
\end{proof}
\begin{lemma}\label{lem:lin-subst-in}
    If $M \Ra D$, and $x$ is linear in $N$, then $N[M/x] \Ra N[D/x]$.
\end{lemma}
\begin{proof}
    By induction on the well-formedness of $N$.
\end{proof}

\noindent
Lemmas~\ref{lem:lin-subst-out} and~\ref{lem:lin-subst-in} are analogous to both
statements of~\cite[Lemma 3.1]{Simpson}. Armed with both, we can
prove the following theorem, which implies the diamond property.
\begin{thm}
    \label{thm:llin-diam}
    If $D \La M \Ra E$ then there exist
    $C, C'$ such that $D \Pb C$
             and      $E \Pb C'$
            with $C \ssim C'$.
\end{thm}
\begin{proof}
  By induction on the shape of $M \Ra D$ and $M \Ra E$.
\end{proof}

\noindent
By this theorem and Corollary~\ref{crit:diam} we conclude that
$\llin$ is confluent, and thus enjoys both UTD and ULD.

%
%

\subsection{A quantum \texorpdfstring{$\lam$}{lambda}-calculus: \texorpdfstring{\qstar}{Q*}}
The \qstar calculus~\cite{DalLago} is a quantum programming language
which models measurement,
an inherently probabilistic operation.
Reduction occurs between \emph{configurations}, which are terms coupled with a
quantum state, and which we will not detail further.
Its semantics does not fix a strategy and, as $\llin$, is also based on \cite{Simpson}.
Quantum variables in \qstar are linear (and not affine),
so they cannot be duplicated nor discarded, as per the no-cloning
\cite{WoottersZurekNATURE82} and no-erasure \cite{PatiBraunsteinNATURE00}
properties of quantum physics.
%
Reduction steps are paired with a \emph{label} describing the reduction
(\eg~which qubit was measured).

The authors prove a property called \emph{strong confluence} which asserts
that any two maximal (possibly infinite) computation trees with a common root
have an equivalent normalized support (that is, the normalized part of both
support distributions are equivalent) and, further, that any normal form
appears in an equal amount of leaves on both trees.

As they note, this property is quite strong,
and not enjoyed by the classical $\lam$-calculus.
Consider the term $(\lam x. \lam y. y) \Omega$.
It has an infinite computation tree without any normal form (as the term
reduces to itself in call-by-value) and also a finite tree with a single $\lam y. y$ leaf,
which of course does not have an equivalent normalized support.
Note that the \qstar well-formedness judgment rejects this term as it is not
linear.

To prove strong confluence, a crucial lemma called \emph{quasi-one-step confluence} is proved,
which is morally a diamond property but with slightly different behaviours
according to the reductions taken.
Reductions are distinguished between two sets, $\N$ and $\K$, and measurements of the form
$\meas_r$.
We will not describe these sets nor \qstar's semantics
(its full description is found in \cite{DalLago}),
and will merely
state the lemma.
The notation $C \ra^p_\alpha D$ means ``$C$ reduces to $D$ with probability
$p$ via the label $\alpha$''; and
$C \ra^p_\N D$ means $C \ra^p_\alpha D$ for some $\alpha \in \N$ (idem $\K$).

\begin{lemma}[Quasi-one-step Confluence for \qstar {\cite[Proposition 4]{DalLagoArXiv}}]\label{lem:qstar-quasi}~

\noindent
    Let $C,D,E$ be configurations and $C \ra_\alpha^p D$, $C \ra_\beta^s E$,
    then:
    \begin{itemize}
        \item If $\alpha \in \K$ and $\beta \in \K$,
            then either $D = E$ or there is $F$ with $D \ra^1_\K F$ and $E \ra^1_\K F$.

        \item If $\alpha \in \K$ and $\beta \in \N$,
            then either $D \ra^1_\N E$ or there is $F$ s.t. $D \ra^1_\N F$ and $E \ra^1_\K F$.

        \item If $\alpha \in \K$ and $\beta = \meas_r$,
            then there is $F$ with $D \ra_{\meas_r}^s F$ and $E \ra^1_\K F$.

        \item If $\alpha \in \N$ and $\beta \in \N$,
            then either $D = E$ or there is $F$ with $D \ra^1_\N F$ and $E \ra^1_\N F$.

        \item If $\alpha \in \N$ and $\beta = \meas_r$,
            then there is $F$ with $D \ra_{\meas_r}^s F$ and $E \ra^1_\N F$.

        \item If $\alpha = \meas_r$ and $\beta = \meas_q$ (with $r \ne q$),
            then there are $t,u \in [0,1]$ and an $F$ such that $pt = su$,
            $D \ra_{\meas_q}^t F$ and $E \ra_{\meas_r}^u F$.
    \end{itemize}
\end{lemma}

From this lemma, the fact that there are no infinite $\K$ sequences,
and a ``probabilistic strip lemma'',
the
authors prove strong confluence~\cite[Theorem 5.4]{DalLago}.

For distribution confluence, a simpler proof can be obtained.
After modelling \qstar as a PARS (roughly using successor
distribution per label) we can readily reinterpret Lemma~\ref{lem:qstar-quasi} to prove
the diamond property for it (by Corollary~\ref{crit:diam}).
From this result, distribution confluence follows,
and therefore also uniqueness of both terminal and limit distributions.
Notably, neither the normalisation requirement for $\K$
nor the probabilistic strip lemma are needed for this fact.

Our obtained distribution confluence is similar, but neither weaker nor
stronger than strong confluence.
It is not weaker as distribution confluence guarantees that
divergences of computations without any normal form can be joined, which
strong confluence does not.
It is also not stronger as it implies nothing of limit distributions that are
not terminal, while it follows from strong confluence that they must
coincide in their normalized part.%
\footnote{It also does not imply the equality
between the amount of leaves on each tree.
This can in fact be recovered by removing the \rn{Split} and \rn{Join}
rules and (straightforwardly) deriving criteria
following
\autoref{sec:Confluence}. One can then conclude
not only that there is
the same amount but
that they are paired with the same weights.}

\section{Conclusions}\label{sec:Conclusions}

We have studied the problem of showing that a probabilistic operational semantics
is not affected by the choice of strategy.
For this purpose, we provided a definition of confluence for probabilistic systems
by defining a classical relation over distributions.
We showed our property of distribution confluence to be appropriate as,
in particular, it implies a uniqueness of terminal distributions, both for
finite and infinite reductions, and gives an equational consistency
guarantee.

We believe this development demonstrates that distribution confluence
provides a reasonable ``sweet spot'' for proving the correctness of
probabilistic semantics, as it provides the expected guarantees about
execution while allowing tractable proofs.
Concretely, the provided proofs for \llin and \qstar are in line with
what one would expect for linear calculi.

%
The proof about \qstar also partially answers the conjecture posed in \cite[Section 8]{DalLago}
(``any rewriting system
enjoying properties like Proposition 4 [our Lemma~\ref{lem:qstar-quasi}] enjoys confluence in the same sense as the one used here'') positively.
The answer is partial since distribution confluence is not strictly equivalent.

Looking ahead, there are several interesting directions to explore.
Firstly, a study of confluence dealing with \emph{terms} (and not just
abstract elements) should provide more insights applicable to
concrete languages. For terms, we expect concepts such as orthogonality to be
of interest.
Secondly, as a generalization, it seems possible to take distribution
weights from any mathematical field and not only the positive reals.
Even if interpreting such systems is not obvious, most of our results
would hold: interestingly, we have not once assumed that weights actually
represent probabilities.
Finally, a quantitative notion of confluence could also be explored,
where a distribution is considered confluent if any divergence of it can
be joined ``up to $\eps$''; in particular, obtaining useful simplified
criteria for said property seems difficult.

\subsection{Related work}


In~\cite{DalLago2017}, similar definitions of evolution and confluence
are introduced.
A subtle yet key difference with ours is that equivalent distributions are
identified and there is no partial evolution.
This means that evolution is not compositional; and in fact the
diamond property over Dirac distributions does not extend over to all
distributions.
This was wrongly stated in an early version \cite{DalLago2017arXiv} of the paper and
subsequently fixed.

In~\cite{Jano2011}, a notion of confluence is defined and proven for an
extension of $\lam_q$~\cite{vanTonderSIAM04} (a quantum $\lam$-calculus) with
measurements.
The proposed property is basically a confluence on computation trees, and the
one we study in this paper is strictly weaker, yet sufficient for UTD and consistency.

In~\cite{DalLago}, already amply discussed, the introduced property is
a strong confluence over maximal trees (either finite or infinite)
which is very related, but neither weaker nor stronger than
distribution confluence.

Finally, in \cite{BournezKirchner} and subsequently in \cite{KirkebyChristiansen},
a property of confluence is defined
and studied over probabilistic rewriting systems which do not contain
any non-determinism (\ie~where $\Ra$ is a partial function).
This is a very different notion of confluence, dealing with punctual
final results instead of distributions, and with different applications.
In this setting, distribution confluence trivially holds.

\paragraph{Acknowledgements}

We thank Exequiel Rivas for reviewing early versions of this work and for many interesting discussions.
We also thank both the anonymous reviewers and LSFA attendees for constructive and encouraging comments.

\bibliographystyle{abbrv}
\bibliography{refs}

\begin{thebibliography}{10}

\bibitem{BournezGarnier}
O.~Bournez and F.~Garnier.
\newblock Proving positive almost-sure termination.
\newblock In J.~Giesl, editor, {\em Proceedings of the 16th International
  Conference on Term Rewriting and Applications (RTA'05)}, volume 3467 of {\em
  Lecture Notes in Computer Science}, pages 323--337. Springer-Verlag, 2005.

\bibitem{BournezHoyrupRTA03}
O.~Bournez and M.~Hoyrup.
\newblock Rewriting logic and probabilities.
\newblock In R.~Nieuwenhuis, editor, {\em Rewriting Techniques and
  Applications}, volume 2706 of {\em Lecture Notes in Computer Science}, pages
  61--75, 2003.

\bibitem{BournezKirchner}
O.~Bournez and C.~Kirchner.
\newblock {Probabilistic Rewrite Strategies. Applications to ELAN}.
\newblock In S.~Tison, editor, {\em Proceedings of the 13th International
  Conference on Rewriting Techniques and Applications (RTA'02)}, volume 2378 of
  {\em Lecture Notes in Computer Science}, pages 252--266. Springer-Verlag,
  2002.

\bibitem{ChurchRosser}
A.~Church and J.~B. Rosser.
\newblock Some properties of conversion.
\newblock {\em Transactions of the American Mathematical Society},
  39(3):472--482, 1936.

\bibitem{DalLago2017arXiv}
U.~Dal~Lago, C.~Faggian, B.~Valiron, and A.~Yoshimizu.
\newblock The geometry of parallelism: Classical, probabilistic, and quantum
  effects.
\newblock arXiv:1610.09629v2, 2016.

\bibitem{DalLago2017}
U.~Dal~Lago, C.~Faggian, B.~Valiron, and A.~Yoshimizu.
\newblock The geometry of parallelism: Classical, probabilistic, and quantum
  effects.
\newblock In {\em Proceedings of the 44th ACM SIGPLAN Symposium on Principles
  of Programming Languages}, POPL 2017, pages 833--845, New York, NY, USA,
  2017. ACM.

\bibitem{DalLagoArXiv}
U.~{Dal Lago}, A.~Masini, and M.~Zorzi.
\newblock Confluence results for a quantum lambda calculus with measurements.
\newblock \href{http://arxiv.org/abs/0905.4567}{arXiv:0905.4567}.
\newblock Extended version of \cite{DalLago}.

\bibitem{DalLago}
U.~{Dal Lago}, A.~Masini, and M.~Zorzi.
\newblock Confluence results for a quantum lambda calculus with measurements.
\newblock In B.~Coecke, P.~Panangaden, and P.~Selinger, editors, {\em
  Proceedings of the 6th International Workshop on Quantum Physics and Logic
  (QPL'09)}, volume 270.2 of {\em Electronic Notes in Theoretical Computer
  Science}, pages 251--261. Elsevier, 2011.

\bibitem{DalLagoZorzi11}
U.~{Dal Lago} and M.~Zorzi.
\newblock Probabilistic operational semantics for the lambda calculus.
\newblock {\em RAIRO Theoretical Informatics and Applications}, 46(3):413--450,
  2012.

\bibitem{DiPierro}
A.~{Di Pierro}, C.~Hankin, and H.~Wiklicky.
\newblock Probabilistic $\lambda$-calculus and quantitative program analysis.
\newblock {\em Journal of Logic and Computation}, 15(2):159--179, 2005.

\bibitem{Jano2011}
A.~D\'iaz-Caro, P.~Arrighi, M.~Gadella, and J.~Grattage.
\newblock Measurements and confluence in quantum lambda calculi with explicit
  qubits.
\newblock In B.~Coecke, I.~Mackie, P.~Panangaden, and P.~Selinger, editors,
  {\em Proceedings of the Joint 5th International Workshop on Quantum Physics
  and Logic and 4th Workshop on Developments in Computational Models
  (QPL/DCM'08)}, volume 270.1 of {\em Electronic Notes in Theoretical Computer
  Science}, pages 59--74. Elsevier, 2011.

\bibitem{Gordon}
A.~Gordon, T.~A. Henzinger, A.~Nori, and S.~Rajamani.
\newblock Probabilistic programming.
\newblock In J.~Herbsleb, editor, {\em Proceedings of the 36th International
  Conference on Software Engineering. Session on Future of Software
  Engineering}, FOSE'14, pages 167--181. ACM, 2014.

\bibitem{Huet}
G.~Huet.
\newblock Confluent reductions: Abstract properties and applications to term
  rewriting systems: Abstract properties and applications to term rewriting
  systems.
\newblock {\em Journal of the ACM}, 27(4):797--821, 1980.

\bibitem{JouannaudKirchnerPOPL84}
J.-P. Jouannaud and H.~Kirchner.
\newblock Completion of a set of rules modulo a set of equations.
\newblock In {\em Proceedings of the 11th ACM SIGACT-SIGPLAN Symposium on
  Principles of Programming Languages}, POPL '84, pages 83--92. ACM, 1984.

\bibitem{KirkebyChristiansen}
M.~Kirkeby and H.~Christiansen.
\newblock Confluence and convergence in probabilistically terminating reduction
  systems.
\newblock In {\em Proceedings of the 27th International Symposium on
  Logic-Based Program Synthesis and Transformation}, LOPSTR '17, pages 33--46.
  ACM, 2017.

\bibitem{Kozen1981}
D.~Kozen.
\newblock Semantics of probabilistic programs.
\newblock {\em Journal of Computer and System Sciences}, 22(3):328--350, 1981.

\bibitem{Martinez17}
G.~Mart\'inez.
\newblock Confluencia en sistemas de reescritura probabilista.
\newblock Master's thesis, Universidad Nacional de Rosario, Argentina, Mar.~27,
  2017.
\newblock Available at \url{https://dcc.fceia.unr.edu.ar/~gmartinez/tesina}.

\bibitem{Newman}
M.~H.~A. Newman.
\newblock On theories with a combinatorial definition of ``equivalence''.
\newblock {\em Annals of Mathematics}, 43(2):223--243, 1942.

\bibitem{PatiBraunsteinNATURE00}
A.~K. Pati and S.~L. Braunstein.
\newblock Impossibility of deleting an unknown quantum state.
\newblock {\em Nature}, 404:164--165, 2000.

\bibitem{PetersonStickelJACM81}
G.~E. Peterson and M.~E. Stickel.
\newblock Complete sets of reductions for some equational theories.
\newblock {\em Journal of the ACM}, 28(2):233--264, 1981.

\bibitem{PlotkinCPS}
G.~D. Plotkin.
\newblock Call-by-name, call-by-value and the $\lambda$-calculus.
\newblock {\em Theoretical Computer Science}, 1(2):125--159, 1975.

\bibitem{SelingerValiron}
P.~Selinger and B.~Valiron.
\newblock A lambda calculus for quantum computation with classical control.
\newblock In P.~Urzyczyn, editor, {\em Proceedings of the 7th International
  Conference on Typed Lambda Calculi and Applications (TLCA'05)}, volume 3461
  of {\em Lecture Notes in Computer Science}, pages 354--368. Springer-Verlag,
  2005.

\bibitem{Simpson}
A.~Simpson.
\newblock Reduction in a linear lambda-calculus with applications to
  operational semantics.
\newblock In J.~Giesl, editor, {\em Proceedings of the 16th International
  Conference on Term Rewriting and Applications (RTA'05)}, volume 3467 of {\em
  Lecture Notes in Computer Science}, pages 219--234. Springer-Verlag, 2005.

\bibitem{vanTonderSIAM04}
A.~van Tonder.
\newblock A lambda calculus for quantum computation.
\newblock {\em SIAM Journal on Computing}, 33:1109--1135, 2004.

\bibitem{WoottersZurekNATURE82}
W.~K. Wootters and W.~H. Zurek.
\newblock A single quantum cannot be cloned.
\newblock {\em Nature}, 299:802--803, 1982.

\end{thebibliography}

\end{document}